\documentclass[11pt]{article}
\usepackage[dvips]{graphics}
\usepackage{epsfig,rotating}
\usepackage{amssymb,amsmath,amsthm}
\usepackage{latexsym}
\usepackage{natbib}
\usepackage[usenames]{color}
\usepackage{comment}
\usepackage{tabu}
\usepackage{multirow}
\usepackage{booktabs}
\usepackage{subfigure}
\usepackage{algorithm}
\usepackage{algpseudocode}
\usepackage{enumitem}

\numberwithin{equation}{section}
\renewcommand{\baselinestretch}{1.6}
\newcommand{\be}{\begin{equation}}
\newcommand{\ee}{\end{equation}}
\newcommand{\beaa}{\begin{eqnarray*}}
\newcommand{\eeaa}{\end{eqnarray*}}
\newcommand{\bea}{\begin{eqnarray}}
\newcommand{\eea}{\end{eqnarray}}

\newcommand{\tr}{\operatorname{tr}}

\newcommand{\mX}{\mathcal{X}}
\newcommand{\mZ}{\mathcal{Z}}
\newcommand{\mM}{\mathcal{M}}

\newcommand{\tvec}{\text{vec}}

\newtheorem{theorem}{ \noindent T{\footnotesize HEOREM}}
\newtheorem{prop}{ \noindent P{\footnotesize ROPOSITION}}[section]
\newtheorem{lemma}{ \noindent L{\footnotesize EMMA}}

\newtheorem{remark}{ \noindent R{\footnotesize EMARK}}

\begin{document}

\title{
Hypothesis Testing for High-Dimensional Matrix-Valued Data
}
\author{ Shijie Cui$^{a}$, Danning Li$^{b}$, Runze Li$^{a}$, and  Lingzhou Xue$^{a}$\\
$^{a}$The Pennsylvania State University
and $^{b}$Northeast Normal University
}

\date{}
\maketitle
\renewcommand{\baselinestretch}{1.6}
\begin{abstract}
This paper addresses hypothesis testing for the mean of matrix-valued data in high-dimensional settings. We investigate the minimum discrepancy test, originally proposed by Cragg (1997), which serves as a rank test for lower-dimensional matrices. We evaluate the performance of this test as the matrix dimensions increase proportionally with the sample size, and identify its limitations when matrix dimensions significantly exceed the sample size. To address these challenges, we propose a new test statistic tailored for high-dimensional matrix rank testing. The oracle version of this statistic is analyzed to highlight its theoretical properties. Additionally, we develop a novel approach for constructing a sparse singular value decomposition (SVD) estimator for singular vectors, providing a comprehensive examination of its theoretical aspects. Using the sparse SVD estimator, we explore the properties of the sample version of our proposed statistic. The paper concludes with simulation studies and two case studies involving surveillance video data, demonstrating the practical utility of our proposed methods.
\end{abstract}

\par \vspace{9pt}
\noindent%
\vfill

\section{Introduction}

With the rapid expansion of data in various fields, its underlying structure has become more complex and multifaceted. In contemporary datasets, data is frequently organized as multi-dimensional arrays, such as matrices or tensors, moving beyond the simplicity of vector representations. For instance, colored images from continuous surveillance footage are typically represented as three-dimensional tensors, whereas black-and-white images are conceptualized as matrices. Additionally, in each dimension, these datasets exhibit a characteristic where the number of pixels or features surpasses the sample count, categorizing them as high-dimensional. This shift in data representation underscores the evolving complexity and diversity of data structures in the digital age.

Statistical inference focusing on means has garnered significant interest, though much of the prior work has concentrated on testing the means of random vectors. For vectors of fixed dimension, the Hotelling $T^2$ test proposed by \citet{hotelling1931generalization} as a method based on sample means, stands out as the most prominent. However, when the dimension $p$ of a vector is comparable with the sample size $n$, the power of the Hotelling $T^2$ test is low. Additionally, when $p$ is larger than $n$, the sample covariance matrix becomes singular, making it impractical to compute the Hotelling $T^2$ statistic. To address the high-dimensional mean testing problem, \citet{bai1996effect} and \citet{chen2010tests} proposed $L_2$ type statistics that do not require the estimation of the inverse of the covariance matrix. \citet{tony2014two} proposed $L_\infty$ type statistics that can be more sensitive to extreme values in the mean vector. \cite{yu2023power} proposed the power-enhanced test of high-dimensional mean vectors and expanded the high-power regions of \citet{chen2010tests} and \citet{tony2014two} to a wider alternative space. \cite{li2024power} developed power-enhanced mean tests for high-dimensional compositional data. For a comprehensive review of testing high-dimensional means, readers may refer to \citet{huang2022overview}.

Although statistical inference for high-dimensional means has received significant focus, research concerning such inference for multidimensional arrays—like matrices or tensors—remains comparatively sparse. Unlike vectors, matrices and tensors exhibit distinct characteristics, including rank, singular values, and singular vectors, which frequently capture the interest of researchers. Among these attributes, the rank of a matrix stands out as a particularly important feature, offering essential insights into the matrix's structure and properties.

Previous studies by \citet{cragg1993testing} and \citet{cragg1997inferring} proposed a minimum discrepancy type test for testing coefficient matrices of instrumental variables models. \citet{robin2000tests} proposed tests by estimating the eigenvalues of quadratic forms of the matrix, which can avoid the need to estimate the covariance matrix. However, all of these studies rely on restrictive model assumptions and assume that the dimension of the matrix is fixed. Furthermore, they all focus on matrix rank tests for parameter estimates rather than data matrices. In this paper, we propose a new test statistic that is valid for testing the rank of high-dimensional matrix-valued data.

Matrix rank testing finds a compelling application in object detection, a prevalent issue within computer vision aimed at pinpointing specific objects within images. For example, the development of a visual sensor to oversee security measures in mines could significantly mitigate losses attributed to common accidents. A critical function of such technology would be to detect unusual collapses by analyzing sequences of images obtained from ongoing surveillance footage within a designated area. Additional applications extend to monitoring landslides, fire outbreaks, and suspicious activities in particular locales. From a statistical standpoint, each image can be interpreted as a tensor or matrix permeated with random noise. The rank test we propose, designed for high-dimensional matrix-valued data, is adept at identifying objects or features that deviate from their surroundings. Consequently, this enhances the capability to spot and respond to suspicious incidents effectively, offering a substantial benefit to areas like security monitoring and disaster prevention.


We provide a theoretical demonstration that the minimum discrepancy test, effective for low-dimensional matrices, loses efficacy in high-dimensional contexts. To address this challenge, we introduce a novel statistic specifically designed for rank inference in high-dimensional matrix-valued data. Our approach represents a significant advancement in tackling the matrix rank testing problem under high-dimensional settings.

We rigorously establish the theoretical properties of the oracle version of our proposed statistic, demonstrating its asymptotic normality and analyzing its asymptotic power function. To mitigate accumulated empirical error, we develop a novel sparse Singular Value Decomposition (SVD) technique and provide a comprehensive theoretical framework for its use. Building on this sparse SVD approach, we construct a sample version of the statistic and examine its theoretical advantages. Simulation studies and two practical applications involving surveillance video data illustrate the effectiveness of our method. Importantly, this work pioneers the application of matrix rank testing in the context of object detection, representing a notable contribution to the field.

The rest of our paper is structured as follows. Section \ref{Sec2} presents the preliminaries including the minimum discrepancy type test proposed by \citet{cragg1997inferring} and useful notations. In Section \ref{sec4.34/4}, we first show the limitations of the minimum discrepancy test in high-dimensional contexts and then introduce our proposed methods. Section \ref{sec4.53/30} is dedicated to simulation studies, presenting empirical proof of the effectiveness of our proposed statistics. Section \ref{4.620220502} applies our proposed methods to two video surveillance datasets, showcasing the practical applications and potential of our approach. Section \ref{sec804092024} includes some concluding remarks. The proofs and technical details are presented in the appendix.

\section{Preliminaries}\label{Sec2}

\subsection{Background}
We now present a formal introduction to the model considered in this paper. Suppose  $\{\mathcal{X}_i, i=1,\cdots, n$\} is  a sample from the following model:
\begin{equation}\label{eqn431}
\mathcal{X}= \Pi^0+ A\mathcal{Z}B,
\end{equation}
where each $\mathcal{X}_i$ is a random matrix and $\mathcal{Z}$ is the $q\times p$-dimensional random matrix with independent entries \{$z_{ij}$\} with mean 0 and variance 1. Without loss of generality, we assume $q\geq p$. $A$ is a $q\times q $ matrix, $B$ is a $p\times p$ matrix and $\Pi^0$ is a $q\times p$ low rank matrix.   Let $\tvec(C_{s\times t})$ be formed by stacking the columns of a matrix $C\in\mathbb{R}^{s\times t}$ into a vector in $\mathbb{R}^{st\times 1}$. Then $\tvec(\mathcal{X})$ has mean $\tvec(\Pi^0)$ and covariance matrix $\Sigma=\Sigma_1 \otimes \Sigma_2$, where $\Sigma_1= B^TB$ and $\Sigma_2= AA^\top$. Let $\rho(\cdot)$  denote the rank of a matrix  and  $\Theta_K=\{\Pi: \rho(\Pi)\leq K \}$. In this work, it is of interest to test the hypothesis: 
\begin{equation}\label{eqn4.34/4}
H_0: \Pi^0 \in \Theta_{K} \quad \mbox{versus}\quad H_1: \Pi^0\not\in \Theta_{K},
\end{equation}
based on the sample $\{\mathcal{X}_i, i=1,\cdots, n \}$.

\citet{cragg1997inferring} proposed a minimum discrepancy type test to test the rank of the parameter matrix of the instrumental variable model. Although the statistic was proposed under a different setting, it can also be used for testing problem (\ref{eqn4.34/4}) with a small modification, when $q$ and $p$ are of fixed dimension. They proposed a quadratic form statistic with an asymptotic $\chi^2$ distribution. The statistic $n\hat{C}(\bar{\mX}, K)$ can be considered,
where
\begin{eqnarray}\label{eqn4.3.3}
\hat{C}(\bar{\mX}, K)=\min _{\Pi \in \Theta_{K},\tvec(\Pi)=\pi}(\hat{\pi}-\pi)^{\top} \hat{\Sigma}^{-1}(\hat{\pi}-\pi)=\min _{\Pi \in \Theta_{K},\tvec(\Pi)=\pi}||\hat{\Sigma}^{-1/2}(\hat{\pi}-\pi)||^2_2.
\end{eqnarray}
The asymptotic properties of the statistic proposed in the minimum discrepancy test are given in Theorem \ref{theorem44}, which is a restatement of Theorem 1 of \citet{cragg1997inferring}.

\begin{theorem}\label{theorem44} [Theorem 1 of \citet{cragg1997inferring}]
Assume that there is an estimate of $\Sigma$,  such that $\hat{\Sigma}\overset{a.s.}{\longrightarrow} \Sigma$. When $\rho(\Pi^0)=R$, we have (i)
$n \hat{C}(\bar{\mX}, K)\stackrel{d}{\longrightarrow}\chi^2_{(p-R)(q-R)}$; (ii) $n\hat{C}(\bar{\mX}, K)\rightarrow\infty $ for $K<R$; (iii) $n\hat{C}(\bar{\mX}, K)\leq n\hat{C}(\bar{\mX}, R)$ for $K>R$.
\end{theorem}

Theorem \ref{theorem44}(i) provides the asymptotic distribution of (\ref{eqn4.3.3}) under the null hypothesis. Furthermore, to obtain a good estimator of the rank of $\Pi^0$, Theorem \ref{theorem44}(ii) and \ref{theorem44}(iii) motivate us to consider sequential hypothesis testing. We could begin with a test with the null hypothesis $H_0: \Pi^0\in\Theta_0$ and estimate $R$ as the smallest value of $K$ for which the test does not reject $H_0: \Pi^0\in\Theta_K$. When the dimension is fixed, the sample covariance matrix, denoted by $\hat{\Sigma}$, can converge to $\Sigma$ almost surely.

\subsection{Notations}

Before proceeding, we define some useful notations throughout this paper. For a matrix $C\in \mathbb{R}^{s\times t}$, let $C_{j,:}$ be the vector of the $j$-th row of $C$ and $C_{:,j}$ be the vector of the $j$-th column of $C$.  We say $C=diag(c_1,\cdots,c_m)$ when $C\in\mathbb{R}^{m\times m}$ is a diagonal matrix with $C_{ii}=c_i$ and 0 for other elements of $C$ and $I_m$ is an identity matrix with size $m$.  For a squared matrix $C\in \mathbb{R}^{s\times s}$, $\tr(C)=\sum_{i=1}^sC_{ii}$. The inner product of two matrices $C,D \in\mathbb{R}^{s\times t}$ (or vectors when $t=1$) , is $\left<C,D\right>=\tr(CD^\top)$.   Let $||C||_F=\sqrt{\left<C,C\right>}$ be the Frobenius norm of $C$. $\sigma_i(C)$  means the $i$-th largest singular value of $C$. Specifically, let $\bar{\sigma}(C) $ or $||C||_2$ be the largest singular value of $C$ (or Euclidan norm when $C$ is a vector) and $\underline\sigma(C)$ be the smallest non-zero singular value of $C$. Let $\mathbb{O}^{s\times t}=\{U\in \mathbb{R}^{s\times t}:U^\top U=I_t\}$ be the collection of column orthogonal matrices. Let $\mathcal{F}_{q,p,K}=\{(U,V,\Lambda):U\in \mathbb{O}^{q\times K},V\in \mathbb{O}^{p\times K},\Lambda=diag(\sigma_1,\cdots,\sigma_K), \sigma_1\geq\cdots\geq\sigma_K\geq 0\}$.  Let the true rank of $\Pi^0$ be $R$ so that  singular value decomposition (SVD) of $\Pi^0$ is: $\Pi^0=U^0\Lambda^0 V^{0^\top}=\sum_{k=1}^R\sigma_k\mu_k\nu_k^\top$, where $(U^0,\Lambda^0,V^0)\in\mathcal{F}_{q,p,R}$ and $\sigma_1\geq\cdots\geq\sigma_R>0$. The choice of $(U^0,\Lambda^0,V^0)$ is unique when there are no repeated singular values but is not necessary to be unique if there are repeated singular values.  Let $\bar{\mX}=\frac{1}{n}\sum^n_{i=1}\mathcal{X}_i$, the sample mean of $\mX$ and similarly $\bar{\mZ}=\frac{1}{n}\sum^n_{i=1}\mathcal{Z}_i$.    Let $\hat{\pi}=\tvec(\bar{\mX})$ and $\Pi^0=\tvec(\Pi^0)$.
 $a\gg b$ means $\lim_{n\rightarrow \infty} a/b=\infty$. $[k]$ stands for $\{1,2,\cdots,k\}$ for an integer $k>0$. For a subgaussion random variable $S$, we define the subgaussian norm of $S$, $||S||_{\psi_2}$, as the smallest 
$K$ such that $E[exp(\lambda S)]\leq exp(\dfrac{\lambda^2K^2}{2})$ for all $\lambda$ holds. $X_n\stackrel{d}{\longrightarrow} X$ denotes $X_n$ converges to $X$ in distribution.
The derivative of a vector function (a vector whose components are functions) $y=(y_1,y_2\cdots,y_m)^T$ with respect to an input vector $x=(x_1,x_2\cdots,x_n)^T$, is written (in numerator layout notation) as
$\frac{\partial y}{\partial x}=\left[\begin{array}{cccc}
\frac{\partial y_1}{\partial x_1} & \frac{\partial y_1}{\partial x_2} & \cdots & \frac{\partial y_1}{\partial x_n} \\
\frac{\partial y_2}{\partial x_1} & \frac{\partial y_2}{\partial x_2} & \cdots & \frac{\partial y_2}{\partial x_n} \\
\vdots & \vdots & \ddots & \vdots \\
\frac{\partial y_m}{\partial x_1} & \frac{\partial y_m}{\partial x_2} & \cdots & \frac{\partial y_m}{\partial x_n}
\end{array}\right].$

\section{Methodology}\label{sec4.34/4}

After showing the limitations of the minimum discrepancy test under the high-dimensional setting in Section \ref{sec1.2}, Section \ref{sec4.413/30} introduces a novel test statistic designed for matrix rank testing in high-dimensional matrices, along with an exploration of its oracle properties. In Section \ref{sec4.4.23/30}, we unveil a new sparse SVD methodology aimed at achieving precise estimations of singular vectors for the mean matrix. Section \ref{sec4.4.33.30} discusses the development of plug-in statistics derived from this new singular vector estimation technique, highlighting their significant properties. 

\subsection{Minimum Discrepancy Test with Diverging $q$ and $p$}\label{sec1.2}

Problems arise for the minimum discrepancy test when $q$ and $p$ diverge with $n$. To simplify the discussion, let us consider the case where $\Sigma$ is proportional to an identity matrix, $\Sigma=\sigma_0 I_{qp\times qp}$. Without loss of generality, we assume $\sigma_0=1$, so that we do not need to estimate $\Sigma$. In this section, we demonstrate that the minimum discrepancy test still works when $\max\{q, p\}=o(n^\frac{1}{4})$, but it fails when $q$ and $p$ are larger. 

Let $T=\min\{q,p\}$ and  $\bar\mX=\sum_{k=1}^{T} \hat{\sigma}_k \hat{\mu}_{k} \hat{\nu}_{k}^\top$ denote the   SVD for $\bar\mX$. 
Consider the following test statistic:
\begin{eqnarray}
T_K=\min _{\Pi \in \Theta_K}n\|\bar\mX-\Pi\|^2_F=n\sum_{k=K+1}^{T}\hat{\sigma}_k^2.
\end{eqnarray}
Let $\pi$ denote $vec(\Pi)$, $z$ denote $vec(\mathcal{Z})$.
Following the notations of \citet{cragg1997inferring}, partition (after any needed re-ordering of columns) $\Pi$ into $\Pi_{1}$  and
$\Pi_{2}$, of $p-K$ and $K$ columns. Respectively, we can write
$\Pi=(\Pi_{1},\Pi_{2})$ with $\Pi_1=\Pi_{2} V$, $V$ is of $K\times (p-K)$. Let $\pi_2=\tvec(\Pi_2)$ and $\mu=\left(\begin{array}{c}   
    \tvec(V) \\  
    \pi_2  
  \end{array}\right)$, then any $\pi$ is a function of $\mu$, $\pi=\pi(\mu)$.  Let
\begin{eqnarray}
B(\mu)=\frac{\partial}{\partial \mu} \pi(\mu)=\left(\begin{array}{cc}
I_{(p-k)\times(p-k)} \otimes \Pi_{2} & V^{\top} \otimes I_{q\times q} \\
0 & I_{qk\times qk}
\end{array}\right)
\end{eqnarray}
be the Jacobian matrix of $\pi(\mu)$. Let $\hat{\mu}=\text{arg}\min _{\Pi \in \Theta_{K}} \|\bar\mX-\Pi\|_{F}^{2}= \text{arg}\min_\mu(\hat{\pi}-\pi(\mu))^\top(\hat{\pi}-\pi(\mu))$. 
Similarly, we can write
$\Pi^0=(\Pi_{10},\Pi_{20})$ with $\Pi_{10}=\Pi_{20} V_0$, Let $\pi_{20}=\tvec(\Pi_{20})$ and $\mu_0=\left(\begin{array}{c}   
    \tvec(V_0) \\  
    \pi_{20}  
  \end{array}\right)$.
We impose the following assumptions for deriving the asymptotic distribution of the statistics in the minimum discrepancy test.

\begin{enumerate}[label=A\arabic*.,ref=A\arabic*,series=myexample]
\item\label{ass13/30} Suppose  $\mathcal{Z}_1,\mathcal{Z}_2,...,\mathcal{Z}_n$ are i.i.d. $q\times p$ dimensional random matrices with subgaussion elements and uniform subgaussion norm $K$ such that $E(\mathcal{Z}_1)=0_{q\times p}$, $Var\{vec(\mathcal{Z}^T_1)\}=I_q\otimes I_p$. Write $Z_1=vec(\mathcal{Z}_1)=(z_1,z_2,...,z_{pq})^T$. We assume $E(z_i^4)=3+\Delta,$ and 
\begin{equation}
E(z_{i1}^{\alpha_1}z_{i2}^{\alpha_2}\cdots z_{iq}^{\alpha_q})=E(z_{i1}^{\alpha_1})E(z_{i2}^{\alpha_2})\cdots E(z_{iq}^{\alpha_q})
\end{equation}
for a positive integer $t$ such that $\sum\limits_{l=1}^{t}\alpha_l\leq 8$ and $i_1\neq i_2\neq\cdots\neq i_q$.
\item \label{ass222}$\underline{\sigma}(B(\mu_0))$ is bounded below.
\item \label{A_3} All elements of $\Pi_{20}$ and $V_0$ are uniformly bounded above. 
\end{enumerate}

\emph{Remark.} Assumption A1 is advocated in \citet{bai1996effect} for testing high-dimensional means and in \citet{li2012two} for testing covariance matrices. It does not assume any specific parametric distribution of elements in the random matrix. Assumptions A2 and A3 are required to exclude extreme singular values in the matrix and keep singular values at the same level when the dimension grows.

\begin{theorem}\label{theorem73/30}
Under Assumptions A1- A3 and $max\{q, p\}=o(n^\frac{1}{4})$, suppose $\rho(\Pi^0)=R$, then
$$\dfrac{T_{R}-(q-R)(p-R)}{2(q-R)(p-R)}\stackrel{d}{\longrightarrow} N(0,1).$$
\end{theorem}

Theorem \ref{theorem73/30} demonstrates that the statistic in the minimum discrepancy test is asymptotically normal when the size of the matrix is moderate, i.e., $o(n^{\frac{1}{4}})$. However, 
the statistic cannot effectively control the size or achieve sufficient power when the size of the matrix is comparable to the sample size $n$. This issue is another manifestation of the curse of dimensionality, as errors accumulate when the size of the matrix is high.

\subsection{The Oracle Method and Property}\label{sec4.413/30}
In this section, we propose a new test statistic that can accommodate cases where the matrix dimension, $p$ and $q$, can be comparable with the sample size $n$, without assuming that $A$ and $B$ are identity matrices.

Given the SVD of $\Pi^0=\sum_{k=1}^R\sigma_k\mu_k\nu_k^\top$, the new statistic is based on the observation that under the null hypothesis, $\operatorname{tr}(\Pi^0\Pi^{0^T})=\sum_{k=1}^{K}\sigma_k^2$, but under alternative, $\operatorname{tr}(\Pi^0\Pi^{0^T})>\sum_{k=1}^{K}\sigma_k^2$. It motivates us to estimate $\operatorname{tr}(\Pi^0\Pi^{0^T})$ and $\sum_{k=1}^{K}\sigma_k^2$ respectively. To estimate $ \operatorname{tr}(\Pi^0\Pi^{0^T})$, 
noticing that $E\mathcal{X}_1\mathcal{X}_2^T=\Pi^0\Pi^{0^T}:=\mM_0$,  a U-statistic is considered: $$U_n=\tr(\frac{\sum_{i\neq j} \mathcal{X}_i\mathcal{X}_j^T}{n(n-1)}):=\tr(\hat\mM_0),$$ 
To estimate $\sigma_k^2$, noticing that $E\mu_k^T\mathcal{X}_1V^0V^{0^\top}\mathcal{X}_2^T\mu_k =\sigma_k^2$ with each $k$, 
another U-statistic is considered: $$T_k=\frac{1}{n(n-1)}\sum_{i\neq j}\mu_k^T\mathcal{X}_iV^0V^{0^\top}\mathcal{X}_j^T\mu_k.$$  
Consider the following test statistic:
\begin{eqnarray}
V_{n}=U_n-\sum_{k=1}^{K} T_k.
\end{eqnarray}
\emph{Remark.} An  observation is that $\sigma_k^2$ can be expressed as $E\mu_k^T\mathcal{X}_1\mathcal{X}_2^T\mu_k$. However, when estimating $\sigma_k^2$, one might question why we use $E\mu_k^T\mathcal{X}_1V^0V^{0^\top}\mathcal{X}_2^T\mu_k$ to define $T_k$, rather than the simpler form $E\mu_k^T\mathcal{X}_1\mathcal{X}_2^T\mu_k$. Although the simplified version has the same desirable oracle properties as the current formulation, it is important to note that this equivalence does not hold for the sample version statistics, as discussed in Section \ref{sec4.4.33.30}. This is due to the crucial assumption of sparsity for both $V^0$ and $U^0$, which is essential for the effectiveness of our sample version statistics.

Under an oracle setting where $U^0, V^0, \Sigma_1$, and $\Sigma_2$ are known, we can obtain the asymptotic normality of the proposed test statistic under certain regularity conditions.

\begin{enumerate}[resume*=myexample]
\item\label{ass4}As $n,p,q \rightarrow
\infty, pq=O(n^4),  \tr(\Sigma^4) =o(\operatorname{tr}^2(\Sigma^2)),\|\Sigma\|_2^2=o( \operatorname{tr}(\Sigma^2)).$
\end{enumerate}

\begin{theorem}\label{thm333}
Suppose assumptions \ref{ass13/30} and \ref{ass4} hold. Under null hypothesis , we have
\begin{eqnarray}
G_n:=\dfrac{U_n-\sum_{k=1}^{K}T_k}{\sqrt{ \dfrac{2 \operatorname{tr}(\Sigma^2)}{n^2}}}\stackrel{d}{\longrightarrow}N(0,1), \text{as } n\rightarrow \infty.
\end{eqnarray}
When $R>K$ and $\dfrac{n\sum_{i=K+1}^R\sigma^2_{i}}{\sqrt{2\tr(\Sigma^2)}}\rightarrow \delta$ for a constant $\delta$, we have \[\sup_x|P(G_n\leq x)-P(N(\delta,1)\leq x)|\rightarrow 0,\] 
 and $N(\delta,1)$ is normal distribution with mean of $\delta$ and variance of 1.
\end{theorem}

Theorem \ref{thm333} describes the theoretical properties of our newly proposed statistic in the oracle setting. Under the null hypothesis, the statistic will be asymptotically normal. It has power under the local alternative. When we set $K=0$, the statistic simplifies to $\dfrac{nU_n}{\sqrt{2 \tr(\Sigma^2)}}$, which is equivalent to the statistic proposed by \citet{chen2010tests} for testing the mean of a random vector when we are testing whether the mean matrix is a zero matrix or not.

There are two remaining challenges to address. Firstly, in the formulation of $T_k$, we assume that the singular vectors are known. However, we still need to find an estimator of singular vectors that can converge quickly. Secondly, we need a reliable estimator of $\tr(\Sigma^2)$. We are going to address these two challenges in the following sections.

\subsection{Sparse SVD Estimator and Its Property}\label{sec4.4.23/30}
For an estimator of the singular vectors of $\Pi^0$, one naive idea is to use singular vectors of $\bar{\mathcal{X}}$. However, as error accumulates when the matrix size is comparable to $n$, the size cannot be controlled well. Hence, sparsity of the matrix is required. Assuming that the mean matrix $\Pi^0$ is sparse, with non-zero entries only on $k$ rows and $l$ columns, where $k$ and $l$ are relatively small compared to the number of rows $q$ and number of columns $p$, and zero otherwise, we label the nonzero rows of $\Pi^0$ by $I=\{j:\Pi^0_{j,:}\neq 0\}$ and the nonzero columns by $J=\{j:\Pi^0_{:,j}\neq 0\}$. Without loss of generality, assume $I=[k]$ and $J=[l]$.  We require a sparse SVD method for a good estimate of the sparse singular vectors.

The concept of sparse singular value decomposition (SVD) has been recently introduced and thoroughly investigated. \citet{yang2016rate} developed a two-stage algorithm tailored for the decomposition of sparse and low-rank matrices, offering a detailed analysis of the algorithm's convergence rate. The initial phase of their algorithm eliminates zero rows and columns, facilitating a robust initialization for singular vectors. Subsequently, the algorithm iteratively refines singular vectors through truncation, enhancing convergence accuracy. Extending this methodology, \citet{zhang2019optimal} adapted the approach for application to higher-order tensors and demonstrated the convergence rate of their estimator. Nevertheless, the error variance assumptions in both studies are simplified, assuming that matrices $A$ and $B$ are identities, which implies a constant error variance. Our proposed methodology diverges from these assumptions, necessitating the derivation of unique theoretical properties due to the complexity of generalizing their results to our framework. Moreover, both \citet{yang2016rate} and \citet{zhang2019optimal} presuppose a specified matrix rank and do not address the theoretical implications when the actual rank deviates from the assumption. In contrast, our approach does not require a predefined matrix rank, introducing a technique to estimate it instead. This adaptation affords a more versatile and applicable method for sparse SVD in the analysis of high-dimensional matrices, circumventing the limitations of previous models.

In a distinct approach, \citet{lee2010biclustering} introduced a sparse SVD technique framed within a penalized regression paradigm, employing a grouped folded concave penalty to induce sparsity in the singular vectors. By reconceptualizing the sparse SVD challenge as a regression issue, their methodology leverages the penalty mechanism to foster sparsity among singular vectors. However, a detailed exploration of the theoretical properties of their estimator was not provided. In this section, we adapt their algorithm to our specific context, where the sparse SVD problem pertains to matrices exhibiting a non-identity covariance structure. We delineate the theoretical attributes of our estimator, demonstrating its consistency and elucidating its convergence rate, thereby extending the application and understanding of sparse SVD methodologies in more complex settings.

One motivation of the method is from the fact that the closest rank-$K$ estimation of a matrix $X\in\mathbb{R}^{q\times p}$ in the sense of Frobenius norm is to solve the following problem \citep{eckart1936approximation}:
\begin{eqnarray}\label{eqna2_3}
(\tilde{U},\tilde{V},\tilde{\Lambda})=\arg\min_{(U,V,\Lambda)\in \mathcal{F}_{q,p,K}}\| X-U\Lambda V^T\|^2_F.
\end{eqnarray} 
To get a sparse solution and consistent result, it motivates us to consider unfolded concave penalty terms adding to the loss function. Consider minimizing the following penalized equation with respect to $(U, V, \Lambda)$:
\begin{equation}\label{eqn2_6}
(\tilde{U},\tilde{V},\tilde{\Lambda})=\arg\min_{(U,V,\Lambda)\in\mathcal{F}_{q,p,K}}\|\bar\mX-U\Lambda V^\top\|_F^2+\sum_{i=1}^{q}p_{\lambda_u}(\|U_{i,:}\|_2)+\sum_{j=1}^{p}p_{\lambda_v}(\|V_{j,:}\|_2).
\end{equation}
(\ref{eqn2_6}) is highly related to penalized regression problem. To see this, for fixed $V$,  minimization of (\ref{eqna2_3}) with respect to $\Lambda$ and $U$ is equivalent to minimization of
\begin{eqnarray}\label{eqn2_5} 
\tilde{U}=\arg\min_{U\in \mathbb{R}^{q\times K}}\| \tvec(X)-(V\otimes I)\tvec(U)\|_F^2,
\end{eqnarray}
and by normalizing $\tilde{U}$ with  $\tilde{U}=U \Lambda $.
Therefore, (\ref{eqn2_5}) can be viewed as regression problem with design matrix $V\otimes I_q$, response $\tvec(X)$ and coefficients $\tvec(U)$.

We next study the theoretical properties of the proposed estimation procedure.
Assume that the penalty function $p_{\lambda}(t_0)$ is increasing
and concave in $t_0\in [0,\infty)$, and has a continuous derivative $p'_{\lambda}(t_0)$ with
$p'_{\lambda}(0+)>0$. 
In addition, assume $p^\prime(t_0,\lambda)/\lambda$ is increasing in $\lambda\in(0,\infty)$ and
$p^\prime(0+,\lambda)/\lambda$ is independent of $\lambda$. Further define the local concavity
of the penalty function $p_\lambda$ at vector $u\in\mathbb{R}^{k}$ as
$$\kappa(p_\lambda,u)=\lim_{\epsilon\rightarrow 0^+}\max_{1\leq j\leq k}
\sup_{t_1<t_2\in(|u_j|-\epsilon,|u_j|+\epsilon)}-\frac{p_\lambda^\prime(t_2)-p_\lambda^\prime(t_1)}{t_2-t_1}.$$
The following notation and assumptions are imposed to establish theoretical result. Suppose $(U^0, V^0,\Lambda^0)$ is a rank-$K$ SVD triplet that can minimize $\|\Pi^0-U\Lambda V^\top\|_F$.
Let $l_u=\min_{j\in[k]}\|U^{0}_{j,:}\|_2$, $l_v=\min_{j\in[l]}\|V^{0}_{j,:}\|_2$, the  minimum
signal of rows of $U^0$ and $V^0$ respectively. Define
$\mathcal{N}_{u0}=\{U\in\mathbb{R}^{q\times K}:\|U-U^0\|_F\leq l_u/2, U_{[k]^c,:}=0\}$, $\mathcal{N}_{v0}=\{V\in\mathbb{R}^{p\times K}:\|V-V^0\|_F\leq l_v/2, V_{[l]^c,:}=0\}$. 

\begin{enumerate}[resume*=myexample]
\item\label{assa5} $0<c_0<\sigma_1(\Pi_0)\leq\cdots\leq\sigma_K(\Pi^0)<c_0$, where $c_0$ and $c_1$ are constants independent of $n$.
\item Let $\tilde{A}=A_{[k],:}$, and $\tilde{B}=B_{:,[l]}$. $\tilde\Sigma_1=\tilde{B}^\top\tilde{B}$, $\tilde\Sigma_2=\tilde{A}\tilde{A}^\top$. We assume $\|\tilde{\Sigma}_1\otimes\tilde\Sigma_2\|_2<C$ for constant $C$ independent of $n$.
\item\label{ass777} Elements of $\mZ$ are subguassian random variables with a uniform subguassian norm  $K_0$.
\item\label{assa6} $\lambda_{u}\gg\max\{n^{1/\varpi+\varsigma}\sqrt{k/n}, \sqrt{\log (qp)/n}\}$ with some $\varpi\geq 8$ and some arbitrary small $\varsigma>0$, $p'_{\lambda_u}(l_u)=o((nk)^{-1/2})$, $\max_{u\in\mathcal{N}_{u0}}\kappa(p_{\lambda_u},u)=o(1)$.\\
$\lambda_{u}\gg\max\{n^{1/\varpi+\varsigma}\sqrt{k/n}, \sqrt{\log (qp)/n}\}$ with some $\varpi\geq 8$ and some arbitrary small $\varsigma>0$, $p'_{\lambda_v}(l_v)=o((nk)^{-1/2})$, $\max_{v\in\mathcal{N}_{u0}}\kappa(p_{\lambda_v},v)=o(1)$.
\end{enumerate}

\begin{theorem}\label{thm999}  Suppose $\Pi^0$ is a rank-R matrix and  $( U^0, V^0,\Lambda^0)$ is a rank-$K$ SVD triplet that can minimize $\|\Pi^0-U^0\Lambda^0 V^{0^\top}\|_F$, where $0<K\leq R$ is an integer. Under assumption \ref{assa5}-\ref{assa6}, we can show that with probability tending to 1, there exists a local minimizer $(\hat{U}, \hat{V}, \hat\Lambda)$ of (\ref{eqn2_6}) satisfying: (i) $\hat{U}_{[k]^c,:}=0, \hat{V}_{[l]^c,:}=0.$(ii) $\|\hat{U}\hat{U}^\top-U^0U^{0^\top}\|_F=O_P(\dfrac{\sqrt{k}+\sqrt{l}}{\sqrt{n}})$, $\|\hat{V}\hat{V}^\top-V^0V^{0^\top}\|_F=O_P(\dfrac{\sqrt{k}+\sqrt{l}}{\sqrt{n}}).$
\end{theorem}

Theoretical properties of the proposed estimator of singular vectors are established in Theorem \ref{thm999}, where the rate of convergence and sparsity are shown. The support of nonzero rows and columns can be estimated precisely. Additionally, an iterative algorithm is proposed to solve (\ref{eqn2_6}), which is shown in Algorithm \ref{alg1}. Tuning parameters $\lambda_u$ and $\lambda_v$ are predetermined, but a data-driven method is required to determine these parameters. When the error variance is constant, such as when $A$ and $B$ are identity matrices, the BIC criterion can be used to select tuning parameters, as discussed by \citet{lee2010biclustering}. However, when error variance is not constant, the BIC criterion is not directly implementable. In the following simulation and real data study, a sample is split into training and validation samples to select tuning parameters. Specifically, the sample $\{\mX_i,i=1,\cdots,n\}$ is divided into two parts: $\{\mX_i,i=1,\cdots,n_1\}$ and $\{\mX_i,i=n_1+1,\cdots,n\}$. The first part is used to find an estimator $(\hat U_\lambda,\hat V_\lambda,\hat \Lambda_\lambda)$ of sparse SVD, and the loss on the second part is considered:
\[L(\lambda)=\sum_{i=n_1+1}^n\|\mX_i-\hat U_\lambda\hat\Lambda_\lambda\hat V_\lambda^\top\|_F^2.\]
$\lambda=(\lambda_u,\lambda_v)$ is tuned by minimizing $L(\lambda)$.

\begin{algorithm}[h]
\caption{}\label{alg1}
\begin{algorithmic}
\item \textbf{Input:} Sample mean matrix $\bar{\mX}\in \mathbb{R}^{q\times p}$, rank $K$ and hyper-parameter $\lambda_u$, $\lambda_v$.
\item \textbf{Output:} Estimators $\hat{U}$, $\hat{V},\hat\Lambda$.
\item[Step 1.] Apply the standard SVD to $\bar{\mX}$. Let $\{\hat\Lambda^{old},\hat{V}^{old},\hat{U}^{old}\}$  denote the first rank-$K$ SVD triplets.
 \item[Step 2.] Set $\tilde{u}=\arg\min_{u\in \mathbb{R}^{qK}}\|\tvec(\bar\mX)-( \hat{V}^{old}\otimes I_q)u\|_F^2+\sum_{i=1}^{q}p_{\lambda_u}(\|U_{i,:}\|_2)$, with $\tvec(U)=u.$
 \item[Step 3.]Reformulate $\tilde{u}$ to a $q\times K$-matrix $\hat{U}^{thr}$ such that $\tvec(\hat{U}^{thr})=\tilde{u}$.  
\item[Step 4.] Orthonomalization with QR decomposition: $\hat{U}^{new}R_U=\hat{U}^{thr}$.
\item[Step 5.] Set $\tilde{v}=\arg\min_{v\in \mathbb{R}^{Kp}}\|\tvec(\bar\mX)-(I_p\otimes \hat{U}^{new})v\|_F^2+\sum_{i=1}^{p}p_{\lambda_v}(V_{i,:}),$ with $\tvec(V^\top)=v$. 
 \item[Step 6.]Reformulate $\tilde{v}$ to a $p\times K$-matrix $\hat{V}^{thr}$ such that $\tvec((\hat{V}^{thr})^\top)=\tilde{v}$.  
\item[Step 7.] Orthonomalization with QR decomposition: $\hat{V}^{new}R_V=\hat{V}^{thr}$.
\item[Step 8.] Set $\hat{V}^{old}=\hat{V}^{new},$  and repeat Steps 2 to 8  until convergence.
\item Set  $\hat{U}=\hat{U}^{new},\hat{V}=\hat{V}^{new}$ at convergence. $\hat{\Lambda}=\hat U^\top\bar{\mX}\hat V$.
\end{algorithmic}
\end{algorithm}

\subsection{The Plug-in Estimator}\label{sec4.4.33.30}
Define the plug-in estimator of squared sum of singular values 
\[\hat{T}_k=\dfrac{\sum_{i\neq j}\tr(\mX_i\hat{V}\hat{V}^\top\mX_j^\top\hat{\mu}_k\hat{\mu}_k^\top)}{n(n-1)}, \sum_{k=1}^K\hat{T}_k=\dfrac{\sum_{i\neq j}\tr(\mX_i\hat{V}\hat{V}^\top\mX_j^\top\hat{U}\hat{U}^\top)}{n(n-1)},\]
where $\hat{U}=(\hat{\mu}_1,\cdots,\hat\mu_K)$ and $\hat{V}=(\hat{\nu}_1,\cdots,\hat\nu_K)$ are sparse estimators of the first $K$ left and right singular vectors of $\Pi^0$ respectively. We require $\hat U$ and $\hat V$ to satisfy the following assumptions.

\begin{enumerate}[resume*=myexample]
\item\label{ass83/30} Let $\hat{I}$  and $\hat{J}$ be the set of nonzero rows of $\hat{U}$ and $\hat V$ respectively. We assume $P([k]\subset\hat{I},|I|\leq \kappa_n)\rightarrow 1$ and $P([l]\subset\hat{J},|J|\leq \kappa_n)\rightarrow 1$. $\kappa_n=o(n)$.
\item\label{ass93/30} $\|\hat{U}\hat{U}^\top-U^0U^{0^\top}\|_F=o_P(\sqrt{\dfrac{\tr(\Sigma^2)}{n^2}}),\|\hat{V}\hat{V}^\top-V^0V^{0^\top}\|_F=o_P(\sqrt{\dfrac{\tr(\Sigma^2)}{n^2}})$.
\end{enumerate}
To estimate $\tr(\Sigma^2),$ motivated by  Theorem 2 from \citet{li2012two},
 we can use 
 \begin{eqnarray*}
  \hat{S}_{n1}^{2}&=&\dfrac{\sum_{i< j}\tr(2\mathcal{X}_i\mathcal{X}_j^T)^2}{n(n-1)} -
 \dfrac{12\sum_{i<j<k}
 \tr(\mathcal{X}_i\mathcal{X}_j^T)\tr(\mathcal{X}_j\mathcal{X}_k^T)}{n(n-1)(n-2)}
 + \dfrac{24\sum_{i<j<k<l}
 \tr(\mathcal{X}_i\mathcal{X}_j^T)\tr(\mathcal{X}_k\mathcal{X}_l^T)}{n(n-1)(n-2)(n-3)}\\
 &=&\dfrac{6\sum_{i<j<k<l}\tr(\mX_i-\mX_j)(\mX_k-\mX_l)^\top \tr(\mX_i-\mX_j)(\mX_k-\mX_l)^\top}{n(n-1)(n-2)(n-3)}.
 \end{eqnarray*}
 It can be shown to be a ratio consistent estimator of $\tr(\Sigma^2).$
 
Base on the above, the oracle statistic can have a sample version $\hat{G}_n:=\dfrac{U_n-\sum_{k=1}^{K}\hat{T}_k}{\sqrt{2\hat{S}_{n1}^{2}/n^2 }}$. The following theorem shows the results that $\hat{G}_n $  can keep the size and power of the oracle statistic.

\begin{theorem}\label{thm55}
Suppose $(U^0,V^0,\Lambda^0)$ is a rank-K SVD triplet of $\Pi_0$ and $(\hat{U},\hat{V},\hat\Lambda)$ is a sparse estimator of it.  Suppose that assumption \ref{ass13/30}, \ref{ass4}, \ref{ass83/30} and \ref{ass93/30} hold.
Under null hypothesis, we have
$$\hat{G}_n\stackrel{d}{\longrightarrow}N(0,1), \text{as } n\rightarrow \infty.$$
When $R>K$ and $\dfrac{n\sum_{i=K+1}^R\sigma^2_{i}}{\sqrt{2\tr(\Sigma^2)}}\rightarrow \delta$ for a constant $\delta$, we have
\begin{eqnarray}\label{eqn4418}
\sup_x|P(\hat{G}_n\leq x)-P(N(\delta,1)\leq x)|\rightarrow 0.
\end{eqnarray}
\end{theorem} 
In Theorem \ref{thm55}, we show that sample version statistic can perform as well as statistics with when estimators of singular vectors have fast convergence rate. The converging rate of  $\hat{U}$ (as well as $\hat{V}$)  must be in an order that $\|\hat{U}\hat{U}^\top-U^0U^{0^\top}\|_F=o_P(\sqrt{\tr(\Sigma^2)}/n)$. Notice that sparse SVD estimator $\hat U_\lambda,\hat V_\lambda$ solved from (\ref{eqn2_6}) can have  a rate of $\|\hat{U}_\lambda\hat{U}_\lambda^\top-U^0U^{0^\top}\|_F=O_P(\dfrac{\sqrt{k}+\sqrt{l}}{\sqrt{n}}).$ Therefore conditions in Theorem \ref{thm55} can be satisfied when $(k+l)n=o_P(\tr(\Sigma^2))$. It works when $p$ and $q$ are large enough. Our following empirical studies and real data analysis are based on this sample version statistic.

\section{Simulation Study}\label{sec4.53/30}
In this section, we show size and power performance of our proposed statistic. We compare $\hat{G}_n$ with minimum discrepancy test statistic $T_K$ reviewed in Section \ref{sec1.2}. $z_{i,j}'s$ follow standard normal distribution for Model (a) and Model (b) and follow $t_2$ distribution for Model (c). For minimum discrepancy test, we replace $\Sigma$ with $\sigma_0I_{qp\times qp}$ and estimate $\sigma^2_0$ by $\hat\sigma^2_0=\dfrac{1}{nqp}\sum\limits_{i=1}^n\sum\limits_{j=1}^q\sum\limits_{k=1}^p(\mX_{i,j,k}-\bar{X}_{j,k})^2$, where $\bar{X}_{j,k}=\dfrac{1}{n}\sum\limits_{i=1}^n\mX_{i,j,k}$.

\begin{table}[htbp]
\caption{Empirical sizes and powers of the test statistics}\label{tabb1a}
\setlength{\tabcolsep}{1.3mm}
\begin{center}
\begin{tabular}{ccccccccccccc}
\toprule
$n$ & $q$ & $p$  & \multicolumn{5}{c}{----------------$(i)$-----------------} & \multicolumn{5}{c}{----------------$(ii)$-----------------}\\
&&&size&\multicolumn{4}{c}{----------power----------}&size&\multicolumn{4}{c}{----------power----------}\\
\multirow{7}{*}{50} & \multirow{4}{*}{50} &  & $c=0$ & 0.1 & 0.2 & 0.3 & 0.4 & 0 & 0.1 & 0.2 & 0.3 & 0.4 \\ \cline{4-13} 
 & & 50 & 0.08&0.08&0.08&0.27&0.63 &0.03&0.03&0.08&0.28&0.63\\
 &  & 100 & 0.04&0.05&0.11&0.40&0.91 &0.01&0.02&0.11&0.42&0.91 \\
 &  & 200 & 0.05&0.06&0.17&0.73&0.98&0.01&0.03&0.15&0.71&0.98\\ \cline{2-13} 
 & \multirow{2}{*}{100} & 100 & 0.06&0.06&0.19&0.72&0.98&0.02&0.02&0.16&0.72&0.98\\
 &  & 200 & 0.04&0.05&0.34&0.97&1&0.01&0.02&0.32&0.97&1\\ \cline{2-13} 
 & 200 & 200 &  0.04&0.07&0.57&1&1&0&0.06&0.56&1&1\\ \toprule
\multirow{6}{*}{100} & \multirow{3}{*}{50} & 50 & 0.04&0.04&0.22&0.75&0.99&0.01&0.03&0.23&0.79&1\\
 &  & 100 &0.04&0.04&0.37&0.96&1&0.01&0.03&0.37&0.96&1\\
 &  & 200 & 0.03&0.04&0.6&1&1&0&0.02&0.61&1&1\\ \cline{2-13} 
 & \multirow{2}{*}{100} & 100 &0.04&0.05&0.93&0.99&1&0.01&0.04&0.99&1&1  \\
  &  & 200 &0.04&0.07&1&1&1&0&0.03&1&1&1 \\\toprule
\end{tabular}
\end{center}
\end{table}

Model (a):
 Set $\Pi^0_{i,j}=1$ when $1\leq i\leq \dfrac{q}{10}, 1\leq j\leq \dfrac{p}{10}$,   $\Pi^0_{i,j}=c$ when $\dfrac{q}{10}< i\leq \dfrac{q}{5},\dfrac{p}{10}< j\leq \dfrac{p}{5}$, and 0 otherwise.  
$\Sigma_1=\left(\sigma_{i j}\right)_{q \times q}$ and $\Sigma_2=\left(\sigma_{i j}\right)_{p \times p}$ respectively, with $\sigma_{i j}=0.25^{|i-j|}$.
The following null and alternative hypotheses are considered:
\begin{eqnarray}\label{eqn81}
H_0: \Pi^0\in\Theta_1 \text{\ \ versus\ \ } H_1: \Pi^0\notin\Theta_1.
\end{eqnarray}
When $c=0$, it corresponds to the null hypothesis. When $c\neq 0$, it corresponds to the alternative hypothesis. Table \ref{tabb1a} depicts the empirical size and power of
the tests under this model, with nominal level $\alpha=0.05$. $\hat{G}_n$ and $T_K$ can have similar power when $c$ is large, but $\hat{G}_n$ can control the size well but $T_K$ cannot. The empirical size of $T_K$ largely deviates from the nominal level, especially when $q$ and $p$ are large compared with $n$.

Model (b):  Set $\Pi^{0}_{1,1}=10, \Pi^{0}_{2,2}=c$, and 0 otherwise. 
$\Sigma_1=\left(\sigma_{i j}\right)_{q \times q}$ and $\Sigma_2=\left(\sigma_{i j}\right)_{p \times p}$ respectively, with $\sigma_{i j}=0.75^{|i-j|}$. We still consider hypotheses (\ref{eqn81}). Results are shown in Table \ref{tab3}. Scenario(i) is the performance of $\hat G_n$, while scenario (ii) is for $T_K$.
We can see that $T_K$ cannot control the size well under this setting. We also find out that power of $\hat{G}_n$ is smaller when size of matrix $q$ and $p$ grows. It is in line of our expectation that in power function (\ref{eqn4418}),  $\delta$ is proportional to $(\sqrt{\tr(\Sigma)})^{-1}$ and effects the local power. 

\begin{table}[htbp]
\caption{Empirical sizes and powers of the test statistics}\label{tab3}
\begin{center}
\begin{tabular}{ccccccccc}
\toprule
$n$ & $q$ & $p$ & \multicolumn{2}{c}{size} & $c=1$& 2 & 3 & 4 \\ \toprule
\multirow{7}{*}{50} & \multirow{4}{*}{50} &  & $(i)$ & $(ii)$ & $(i)$ & $(i) $ & $(i)$  & $(i)$  \\ \cline{4-9} 
 & & 50 &0.06&0.26&0.06&0.13&0.43&0.91  \\
 &  & 100 & 0.06&0.28&0.06&0.09&0.24&0.61  \\
 &  & 200 & 0.05&0.26&0.05&0.05&0.13&0.3 \\ \cline{2-9} 
 & \multirow{2}{*}{100}& 100 & 0.06&0.32&0.06&0.09&0.16&0.37  \\&&200&0.06&0.35&0.05&0.06&0.10&0.25\\ \toprule
\multirow{6}{*}{100} & \multirow{3}{*}{50} & 50 & 0.04&0.31&0.07&0.38&0.95&1\\
 &  & 100 & 0.04&0.32&0.06&0.19&0.71&0.99\\
 &  & 200 & 0.05&0.33&0.04&0.14&0.41&0.88\\ \cline{2-9} 
 & \multirow{2}{*}{100} & 100&0.05&0.34&0.05&0.11&0.46&0.89\\
 &&200&0.04&0.33&0.05&0.09&0.36&0.66\\
\toprule
\end{tabular}
\end{center}
\end{table}

\begin{remark}
Model (c):  To show the robustness of our statistics, we tried our statistics when setting the noise to $t_2$ distribution. The other settings are the same with Model (b). Table \ref{tab4} shows the empirical sizes and powers of our test statistics and minimum discrepancy test. Again, the minimum discrepancy test can't control the size well while our statistic can control the size well and has power when the parameter deviates from the null hypothesis.

\begin{table}[htbp]
\caption{Empirical sizes and powers of the test statistics}\label{tab4}
\begin{center}
\begin{tabular}{ccccccccc}
\toprule
$n$ & $q$ & $p$ & \multicolumn{2}{c}{size} & $c=1$& 2 & 3 & 4 \\ \toprule
\multirow{7}{*}{50} & \multirow{4}{*}{50} &  & $(i)$ & $(ii)$ & $(i)$ & $(i) $ & $(i)$  & $(i)$  \\ \cline{4-9} 
 & & 50 &0.07&0.29&0.08&0.15&0.39&0.90  \\
 &  & 100 & 0.06&0.29&0.08&0.11&0.25&0.63  \\
 &  & 200 & 0.06&0.35&0.07&0.09&0.15&0.31 \\ \cline{2-9} 
 & \multirow{2}{*}{100}& 100 & 0.04&0.27&0.06&0.10&0.19&0.42  \\&&200&0.06&0.35&0.08&0.10&0.13&0.28\\ \toprule
\multirow{6}{*}{100} & \multirow{3}{*}{50} & 50 & 0.05&0.26&0.09&0.42&0.98&1\\
 &  & 100 & 0.04&0.36&0.09&0.22&0.74&1\\
 &  & 200 & 0.05&0.30&0.06&0.17&0.45&0.90\\ \cline{2-9} 
 & \multirow{2}{*}{100} & 100&0.06&0.24&0.08&0.14&0.47&0.91\\
 &&200&0.07&0.31&0.09&0.12&0.41&0.70\\
\toprule
\end{tabular}
\end{center}
\end{table}
\end{remark}

\section{Real Data Analysis}\label{4.620220502}
\subsection{Detection of Pedestrians}\label{4.5.120220502}
We demonstrate the applicability of our proposed method by applying it to a dataset from a video surveillance system that monitors pedestrians passing through a specific area. This dataset has also been used in the evaluation of various motion detection algorithms in the field of computer vision, as described in \citet{wang2014cdnet}.
\begin{figure}[!h]
  \centering
  \includegraphics[width=14cm]{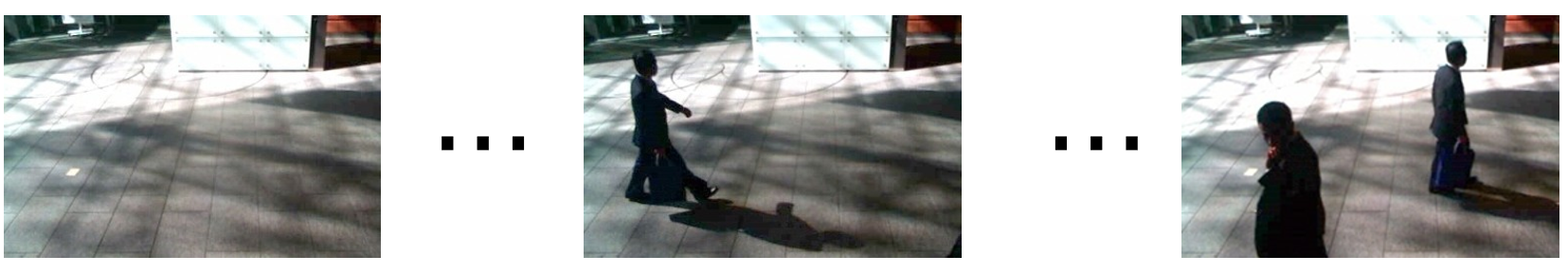}\\
  \caption{Sample video frames at different time. Left: background. Middle: one person. Right: two persons}\label{videofig}
\end{figure}

The data set we analyze comes from a video surveillance system monitoring pedestrians passing through a specific area, which has been used in the literature to compare motion detection algorithms in computer vision (see \citet{wang2014cdnet} for details). The data consists of 1105 sequential images, each taken at a different time point (see Figure \ref{videofig}). We use the gray-level matrix of each image as input and aim to detect changes in matrix rank at different time points using our proposed inference method and select images with higher rank. To achieve this, we create a sliding window, where $n=10$ consecutive images are treated as one sample, and we move the window back by $5$ time units each time. In total, we have 220 samples, and each image has a resolution of $61\times95$, making it a high-dimensional problem where $q\times p$ is much greater than $n$. We perform a sequential hypothesis testing for matrix rank on each sample, where we test $H_0: \rho(\Pi^0)\leq K$ vs. $H_1:\rho(\Pi^0)> K$ for $K=1,\cdots,T$, where $T$ is a relatively large integer, and the rank is determined by the smallest $K$ that cannot reject the null hypothesis.

If the rank of the matrix is tested to be 0 for a particular sample, we conclude that there is no object other than the background in that sample. Otherwise, we have confidence that there are objects present. We compare the results with the ground truth and find that our method can well separate the cases when there are people or not. The false-positive rate of our method is 0.078, and the false-negative rate is 0.059. In addition to detecting the presence of people, our rank detection result shows clear patterns: rank increases when people enter the area or more people join, and decreases when people leave the area. This provides additional insights into the dynamics of pedestrian movement in the monitored area.

\begin{figure}
	\centering
	\subfigure{
		\begin{minipage}[b]{0.3\textwidth}
			\includegraphics[width=1\textwidth]{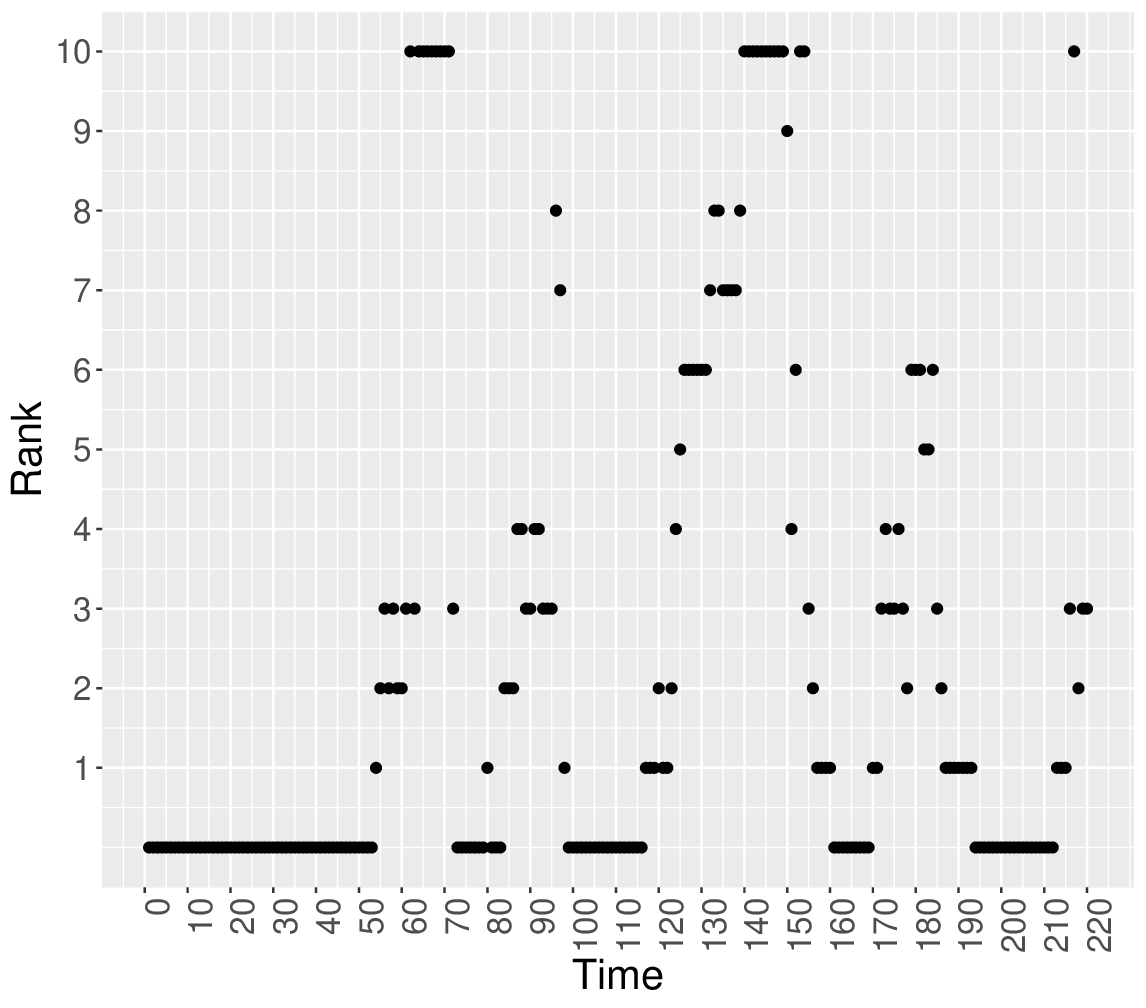} 
		\end{minipage}
	   }
    	\subfigure{
    		\begin{minipage}[b]{0.3\textwidth}
   		 	\includegraphics[width=1\textwidth]{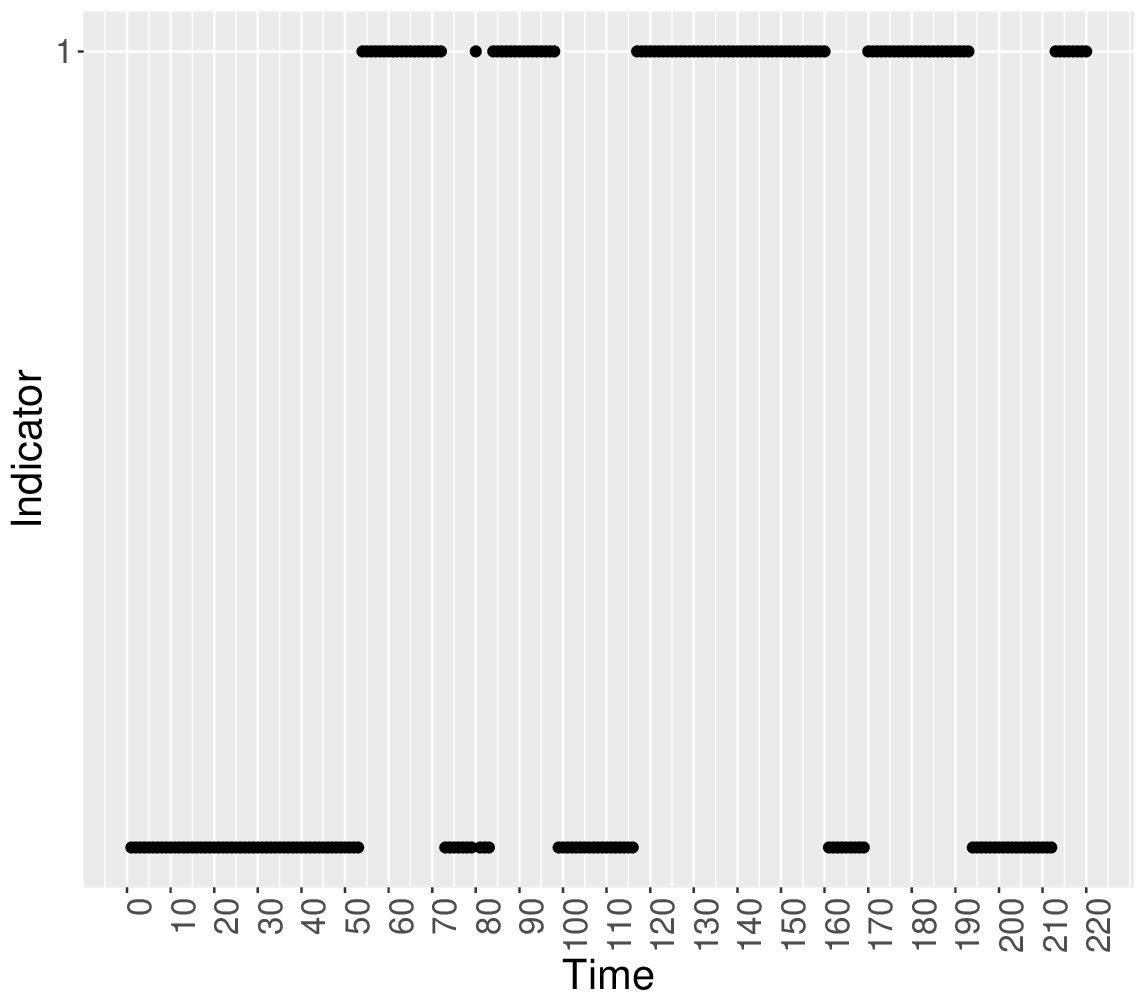}
    		\end{minipage}
    	}
	\subfigure{
		\begin{minipage}[b]{0.3\textwidth}
			\includegraphics[width=1\textwidth]{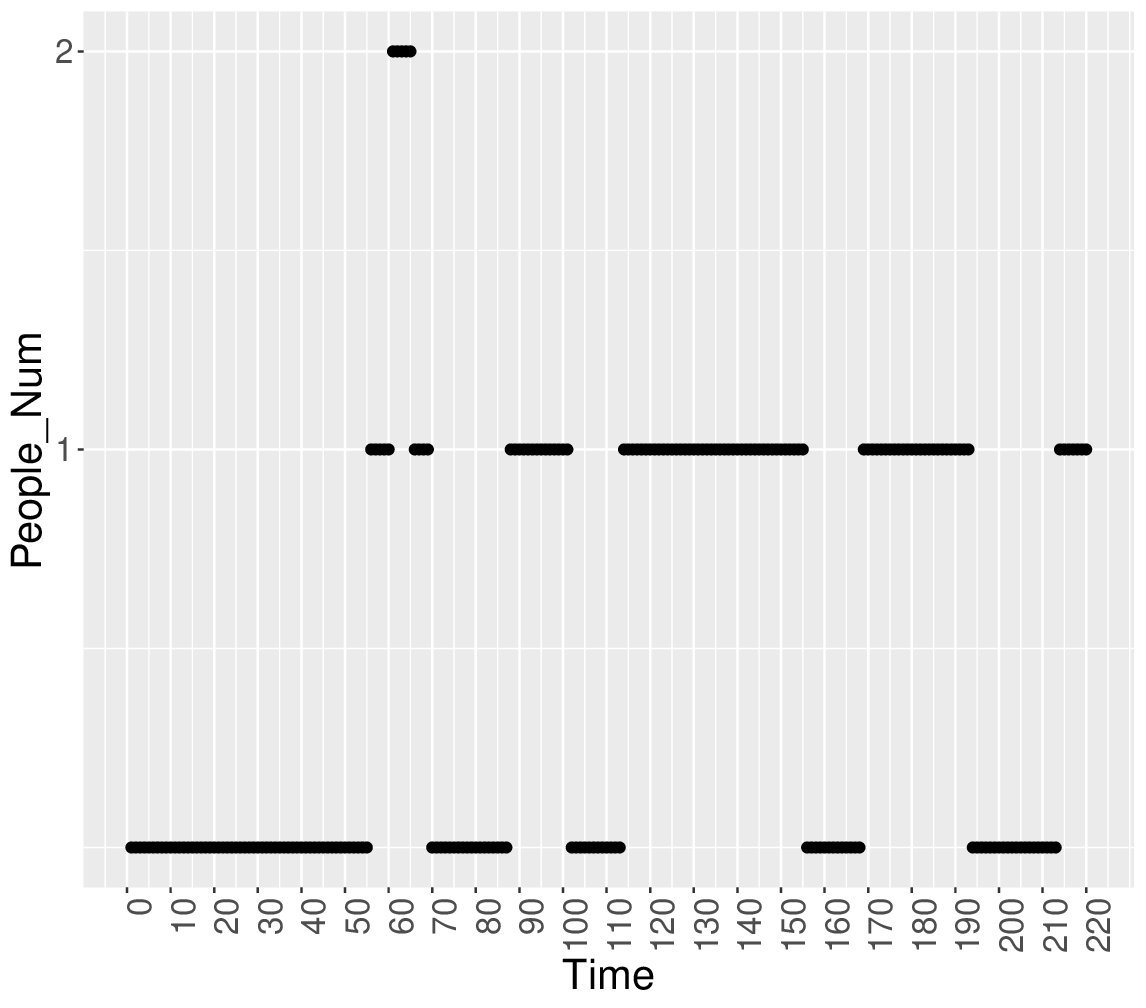} 
		\end{minipage}
		\label{fig:grid_4figs_1cap_4subcap_3}
		}
	\caption{Left: matrix rank detected. Middle: matrix rank greater than 0 or not. Right: ground truth of number of people in the video}\label{fig4}
\end{figure}
\begin{figure}
	\centering
	\subfigure{
		\begin{minipage}[b]{0.4\textwidth}
			\includegraphics[width=1\textwidth]{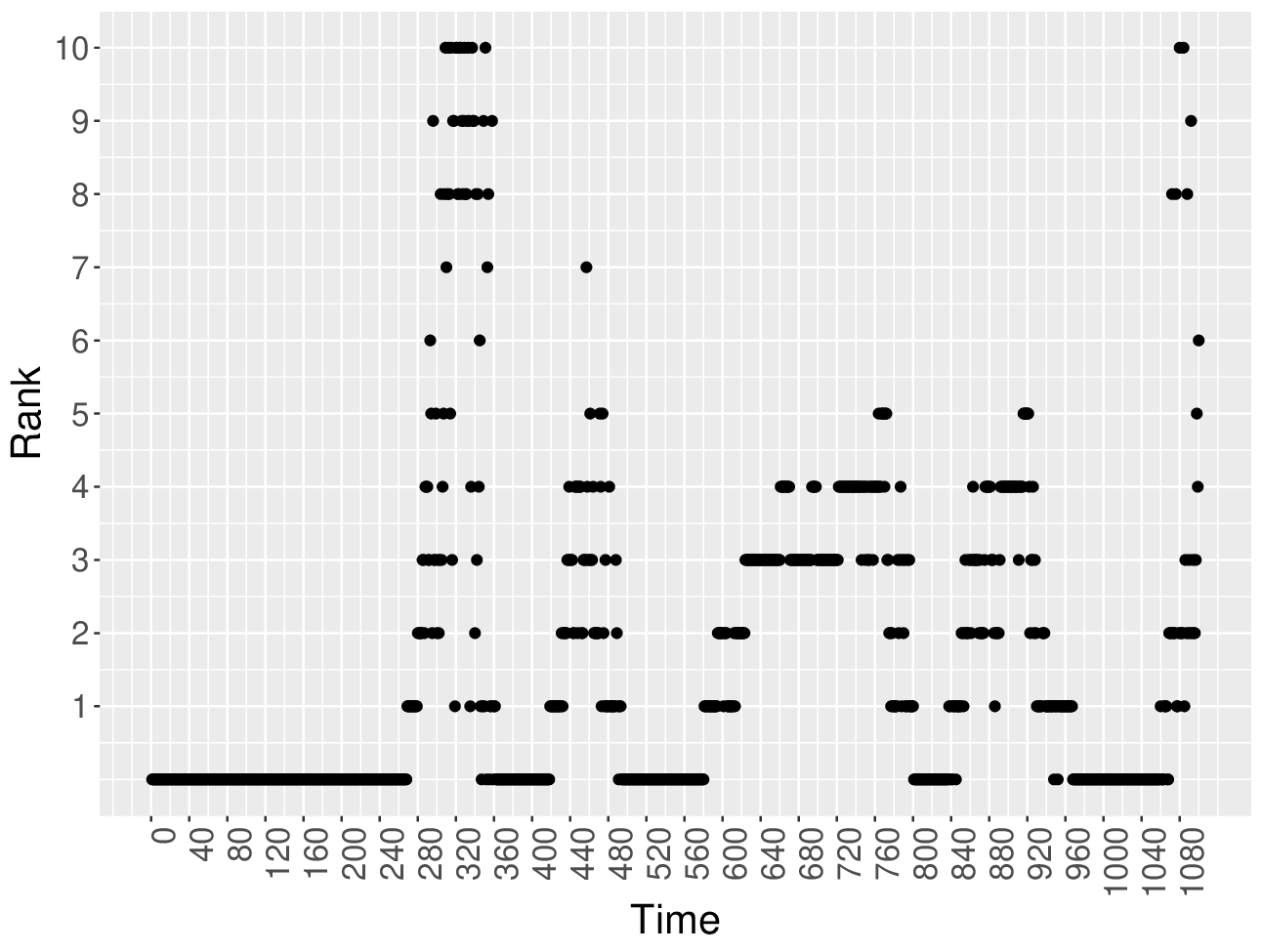} 
		\end{minipage}
	   }
    	\caption{Matrix rank detected}\label{Figure4.4}
\end{figure}

We reduce the sample size to $n=5$ and the sliding window length to 1. The results are presented in Figure \ref{Figure4.4}.  We compare the results with the ground truth of whether there are people or not in the sequence of images and compute the false-positive rate to be 0.018 and the false-negative rate to be 0.072.

\subsection{Detection of Cars in A Parking Lot}
We have the second dataset, which is a video surveillance recording of a parking lot, as illustrated in Figure \ref{videofigcar}. The number of cars in the parking lot varies across different periods, and we aim to use our rank testing procedure to demonstrate the changes in the number of cars over time.

\begin{figure}[!h]
  \centering
  \includegraphics[width=14cm]{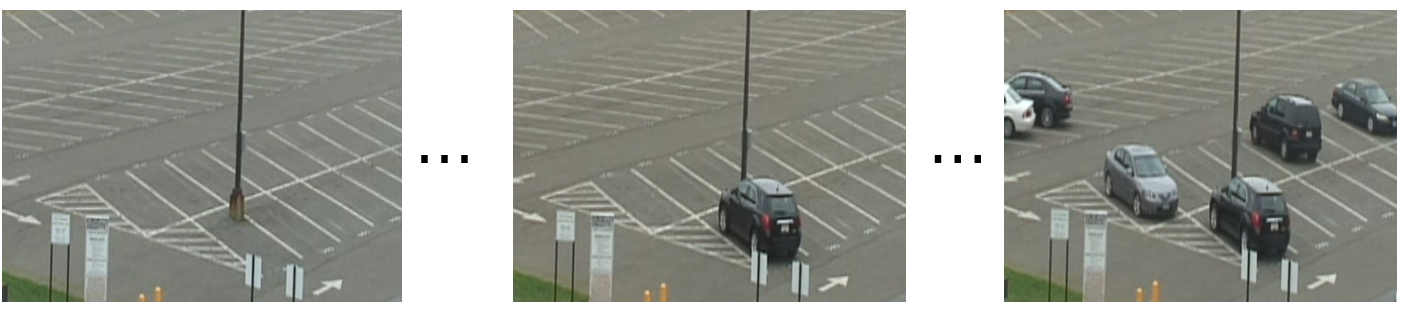}\\
  \caption{Sample video frames at different times. Left: background. Middle: one car. Right: several Cars}\label{videofigcar}
\end{figure}

We use the gray-level matrix of each image as input and have a total of 900 images. We use a sliding window of $n=10$ and move the sliding window 1-time unit back each time. Each image has a resolution of $34\times 81$. We consider the same sequential hypothesis testing problem as in Section \ref{4.5.120220502}. The final results of rank determination of the video data at different time points are shown in Figure \ref{fig4.6202205021913}. Compared to the ground truth, there is always an increase in rank when a car enters the parking lot. The amount of increase in rank for different cars varies, possibly due to the size of the car, which depends on the distance of the car from the camera. Additionally, one car might be blocked by another if they are parked next to each other. The rank often increases gradually when a car enters the parking lot, which can be attributed to the car's movement and entry into the video, as well as the sliding window methodology causing a time lag when a car enters the parking lot.

\begin{figure}
	\centering
	\subfigure{
		\begin{minipage}[b]{0.4\textwidth}
			\includegraphics[width=1\textwidth,height=6cm]{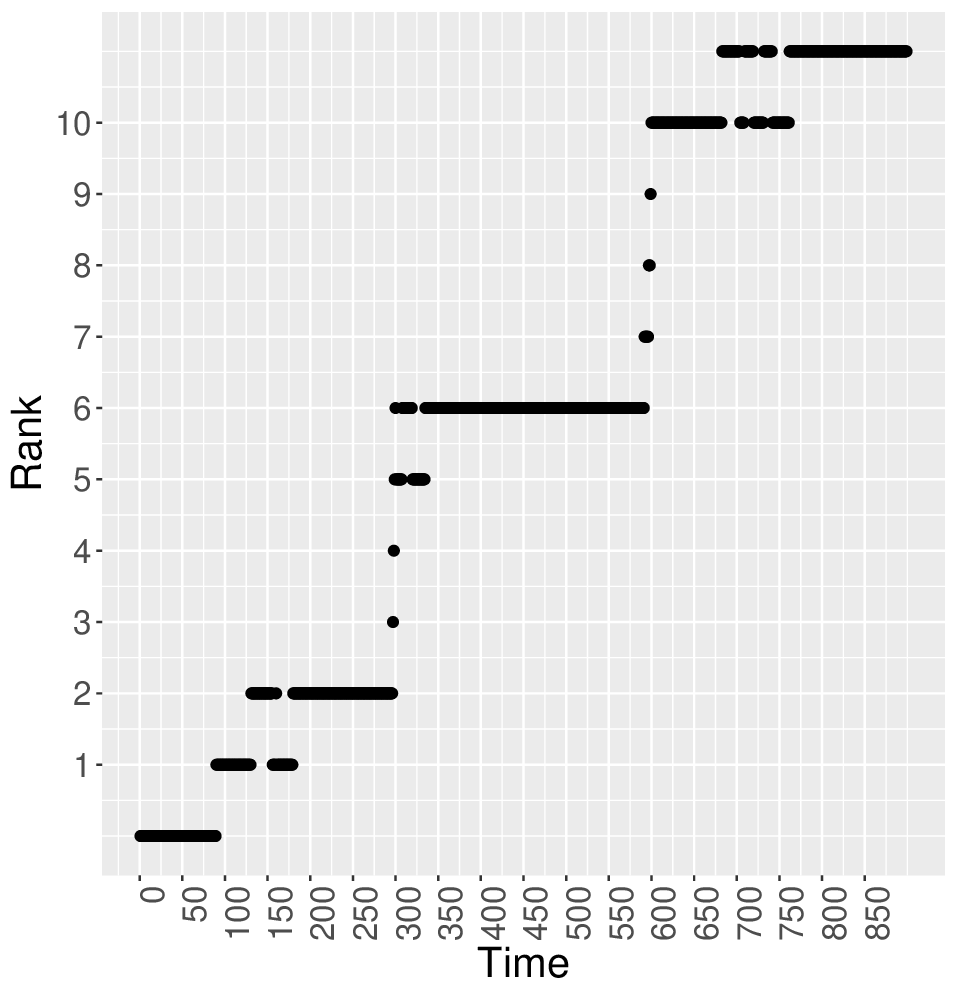} 
		\end{minipage}
	   }
    	\subfigure{
    		\begin{minipage}[b]{0.4\textwidth}
   		 	\includegraphics[width=1\textwidth,height=6cm]{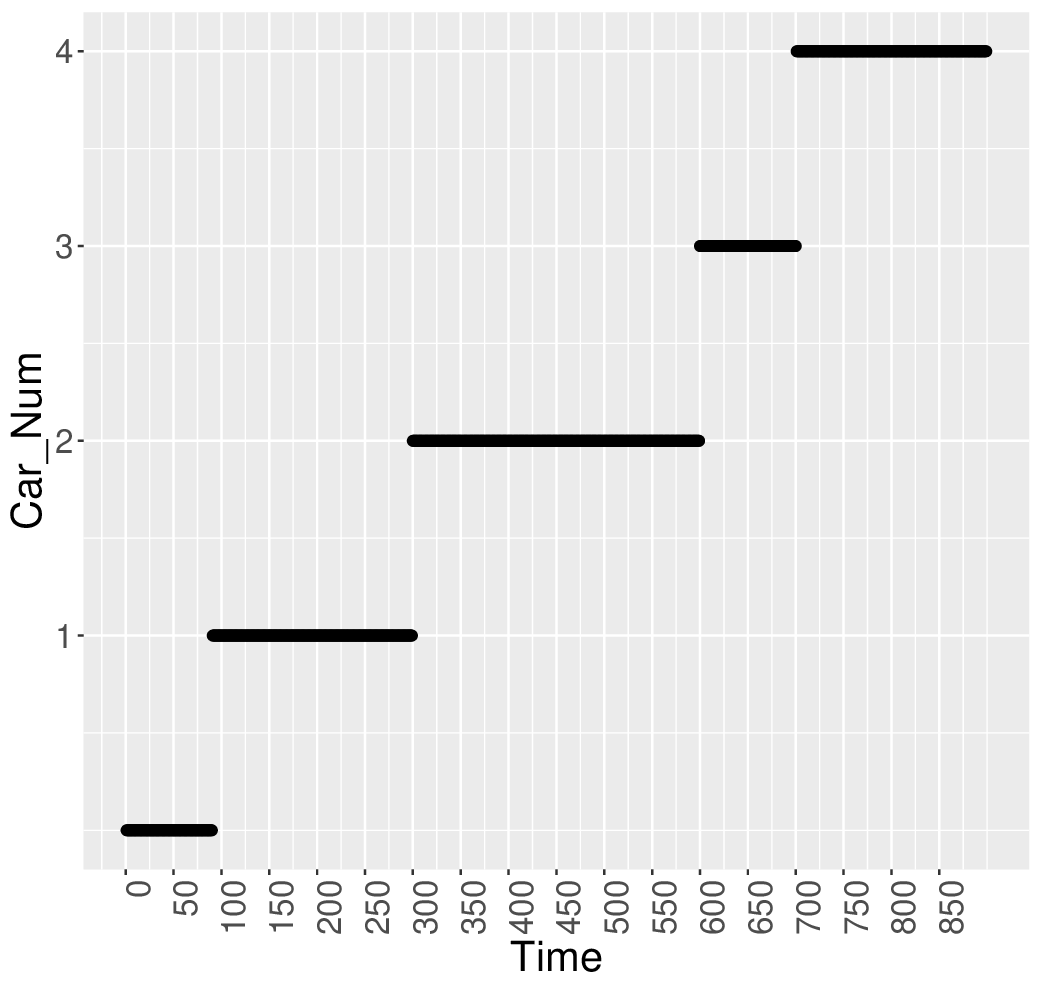}
    		\end{minipage}
    	}
\caption{Left: matrix rank detected. Right: number of cars in the parking lot (Ground Truth)}\label{fig4.6202205021913}
\end{figure}

\section{Conclusion}\label{sec804092024}
In this paper, we propose novel methods for statistical inference of the rank of high-dimensional matrix-valued data. To the best of our knowledge, this is the first work that focuses on testing the structure of matrix-valued data. Our contributions fill the theoretical and practical gap between inference on vector-valued data and inference of matrix-valued data. Specifically, we introduce a new test statistic that is suitable for high-dimensional settings, as the minimum discrepancy type test that works for low-dimensional settings fails under such scenarios. We also prove the oracle property of the proposed test statistic. Additionally, we propose a practical method for sparse SVD. We demonstrate the effectiveness of our testing procedure through the analysis of two surveillance video datasets, which shows the potential application of our method in object detection.

\bibliographystyle{agsm} 
\bibliography{hdtest.bib}

\end{document}